\documentclass[10pt, doublecolumn]{IEEEtran}
\usepackage{epsfig,latexsym}
\usepackage{float}
\usepackage{indentfirst}
\usepackage{amsmath}
\usepackage{bm}
\usepackage{amssymb}
\usepackage{times}
\usepackage{subfigure}
\usepackage{psfrag}
\usepackage{hyperref}
\usepackage{cite}
\usepackage{lastpage}
\usepackage{fancyhdr}
\usepackage{color}
 \usepackage{amsthm}
\usepackage{bigints}
\newtheorem{Lemma}{Lemma}
\newtheorem{lemma}[Lemma]{$\mathbf{Lemma}$}

\newcounter{problem}
\newcounter{save@equation}
\newcounter{save@problem}
\makeatletter
\newenvironment{problem}
{\setcounter{problem}{\value{save@problem}}%
  \setcounter{save@equation}{\value{equation}}%
  \let\c@equation\c@problem
  \subequations
}
{\endsubequations
  \setcounter{save@problem}{\value{equation}}%
  \setcounter{equation}{\value{save@equation}}%
}

\begin{document}%%
\title{ \huge{ Advantages of  NOMA for Multi-User BackCom Networks }}

\author{ Zhiguo Ding, \IEEEmembership{Fellow, IEEE} and  H. Vincent Poor, \IEEEmembership{Life Fellow, IEEE}\thanks{
  Z. Ding and H. V. Poor are  with the Department of
Electrical Engineering, Princeton University, Princeton, NJ 08544,
USA. Z. Ding
 is also  with the School of
Electrical and Electronic Engineering, the University of Manchester, Manchester, UK (email: \href{mailto:zhiguo.ding@manchester.ac.uk}{zhiguo.ding@manchester.ac.uk}, \href{mailto:poor@princeton.edu}{poor@princeton.edu}).

}\vspace{-2em}} \maketitle
\begin{abstract}
Ambient   backscatter communication (BackCom) is faced with the challenge  that   a single BackCom device  can  occupy  multiple orthogonal resource blocks unintentionally. As a result, in order to avoid co-channel interference,    a conventional approach is to serve multiple  BackCom devices   in different time slots, which reduces both spectral efficiency and connectivity. This letter demonstrates that the use of  non-orthogonal multiple access (NOMA) can efficiently improve the system throughput and support massive connectivity in ambient BackCom networks. In particular, two   transceiver design  approaches    are developed  in the letter   to realize   different tradeoffs between system performance and complexity.    
\end{abstract}\vspace{-0.2em}

\vspace{-1.2em} 

\section{Introduction}
Recently, various  novel backscatter communication (BackCom) approaches have been developed  to support   the envisioned ultra-massive machine type communications (umMTC)\cite{you6g, 8476159,8446004,8907447,bacnomamtc}.  The key idea of these BackCom approaches is to use the signals sent by users in a legacy system for exciting  the   circuits of BackCom devices.  This type of  BackCom is featured by the challenge  that    a single BackCom device   can occupy  multiple orthogonal resource blocks unintentionally. For example, if the legacy system is based on    orthogonal frequency-division multiple access (OFDMA), a   BackCom device needs to reflect all     the legacy   signals sent at different subcarriers, which means that a signal sent by the   device  appears at  multiple subcarriers. A similar challenge also presents if   space division multiple access (SDMA) is used in the   legacy system. In order to avoid the co-channel interference between BackCom devices,   orthogonal multiple access (OMA) approaches have been conventionally used, e.g., BackCom devices are served in different time slots \cite{8476159,8446004}, which   reduces both spectral efficiency and connectivity. 

This letter is to demonstrate the advantage of using non-orthogonal multiple access (NOMA) to improve the system throughput and connectivity in BackCom networks. In particular, this letter considers a full-duplex  (FD) network, where downlink users are viewed as legacy users and multiple uplink BackCom devices  are served simultaneously in the downlink users' channels by using BackCom assisted NOMA (BAC-NOMA) \cite{bacnomamtc}. Unlike   \cite{bacnomamtc}, multiple legacy users are considered in this letter, where the use of both SDMA and OFDMA  in the legacy system is  investigated. Two   transceiver design  approaches    are proposed in the letter, where one approach can realize  the sum capacity of the multiple access channel (MAC) and the other one can be implemented in a low complexity manner. The provided simulation results demonstrate that  BAC-NOMA realizes a larger   throughput and support more devices than OMA, even if random resource allocation is used. 

\vspace{-1em}
\section{System Model}
Consider an ambient BackCom network, where SDMA is used in the legacy  system.   In particular, in the legacy system, a base station serves  $K$ downlink users,   denoted by ${\rm U}_k$, via $K$ spatial beamforming vectors, denoted by $\mathbf{w}_k$, $1\leq k \leq K$. The application of BAC-NOMA can ensure that additional  $M$ uplink BackCom devices,   denoted by ${\rm BD}_m$, $1\leq m \leq M$, are admitted to the bandwidth used by the legacy system \cite{bacnomamtc}. It is assumed that each user/device is  equipped with a single antenna,  the base station has $N$ antennas with the full-duplex capability, and each uplink device is equipped with a BackCom circuit. We note that the proposed BAC-NOMA scheme can also been applied  to the case with OFDMA based legacy systems, as shown in Section \ref{section 4}.

 Denote $x_k$ by the   signal sent by downlink user ${\rm U}_k$.  Based on the principle of SDMA,   the base station broadcasts the following superimposed signals: $
 \mathbf{s}_0 = \sqrt{P_0}\sum^{K}_{k=1}\mathbf{w}_kx_k$,
 where $P_0$ denotes the transmit power of the base station \cite{mojobabook}. Optimizing the transmission strategy of the legacy system, i.e., optimizing $\mathbf{w}_k$ and $P_0$, is beyond the scope of this letter, since  it should happen before the uplink BackCom devices are admitted via BAC-NOMA. 
 
 Each downlink user receives the following:
  \begin{align}
y_k = \mathbf{g}_k^T \mathbf{s}_0+\sum^{M}_{m=1}\sqrt{\eta_m} {g}_{m,k}\mathbf{h}_m^T \mathbf{s}_0s_m  + {n}_{k},
 \end{align}
 where $\mathbf{g}_k$ denotes the channel vector between the base station and downlink user ${\rm U}_k$, $\mathbf{h}_m$ denotes the channel vector between the base station and uplink device ${\rm BD}_m$, $g_{m,k}$ denotes the channel gain between ${\rm U}_k$ and  ${\rm BD}_m$, $\eta_m$ denotes the BackCom reflection coefficient of ${\rm BD}_m$ \cite{wongwcnc20}, $s_m$ denotes the signal sent by ${\rm BD}_m$, and $n_k$ denotes the receiver noise. For notational convenience,  it is assumed that the noise terms at different receivers  have the same power,  denoted by $\sigma^2$.

Because $s_m$ is unknown to the downlink users, the term, $I\triangleq \sum^{M}_{m=1}\sqrt{\eta_m} {g}_{m,k}\mathbf{h}_m^T \mathbf{s}_0s_m$, can be treated as interference, whose power is given by  $
\mathcal{E}_{x_k,s_m} \left\{ II^* \right\}
= P_0\sum^{M}_{m=1} {\eta_m} |{g}_{m,k}|^2
  |\mathbf{h}_m^T\mathbf{W}|^2$, 
 where $\mathcal{E}\{\cdot\}$ denotes the expectation operation,   and $\mathbf{W}=\begin{bmatrix}\mathbf{w}_1&\cdots &\mathbf{w}_K \end{bmatrix}$. Therefore, the data rate of downlink user ${\rm U}_k$ is given by
 \begin{align}\nonumber
 &R_{k}^{\rm D} = \log\left( 1+\frac{P_0|\mathbf{g}_k^T\mathbf{w}_k|^2}{P_0\sum^{K}_{i=1,i\neq k}|\mathbf{g}_k^T\mathbf{w}_i|^2 + \mathcal{E}_{x_k,s_m}  \left\{ II^* \right\}  + \sigma^2 }
 \right).
 \end{align}
 
The base station  receives the following observation:  
 \begin{align}\label{overall mac}
\mathbf{y}_{\rm BS} = \sum^{M}_{m=1}\sqrt{\eta_m}\mathbf{h}_m\mathbf{h}_m^T \mathbf{s}_0s_m + \mathbf{s}_{\rm SI} +\mathbf{n}_{\rm BS}
 \end{align}
where    $\mathbf{s}_{\rm SI}$ is assumed to be complex Gaussian distributed, i.e., $\mathbf{s}_{\rm SI}\sim CN(0, \alpha P_0\mathbf{C}_{\rm SI})$,  $\mathbf{C}_{\rm SI}$ denotes the covariance matrix of the self-interference channels and $\alpha$, $0\leq \alpha\leq 1$, indicates the amount of the residual self-interference \cite{7372448}.  

Depending on how   successive interference cancellation (SIC) is carried out, different sum rates can be realized for the BackCom devices, as shown in the following section.  

\section{Two  Approaches  with Different Tradeoffs Between Performance and Complexity }\label{subsection 3}
\subsection{Approach I - A Sum-Capacity Approaching Design}\label{subsection 3.1}
By    applying the noise pre-whitening process \cite{Goldsmith03}, i.e., applying  a detection matrix $\left(\sigma^2\mathbf{I}_N+\alpha P_0\mathbf{C}_{\rm SI}\right)^{-\frac{1}{2}} $ to $\mathbf{y}_{\rm BS}$,  the model in \eqref{overall mac} can be rewritten as follows:
 \begin{align} \label{overall macx1}
\tilde{\mathbf{y}}_{\rm BS} =& \sum^{M}_{m=1}\sqrt{\eta_m}\mathbf{h}_m^T \mathbf{s}_0\tilde{\mathbf{h}}_ms_m + \tilde{\mathbf{n}}_{\rm BS} = \tilde{\mathbf{H}} \mathbf{D}\mathbf{s} + \tilde{\mathbf{n}}_{\rm BS}
 \end{align}
where $\tilde{\mathbf{y}}_{\rm BS} =\left(\sigma^2\mathbf{I}_N+\alpha P_0\mathbf{C}_{\rm SI}\right)^{-\frac{1}{2}} \mathbf{y}_{\rm BS}$,   $\tilde{\mathbf{h}}_m=\left(\sigma^2\mathbf{I}_N+\alpha P_0\mathbf{C}_{\rm SI}\right)^{-\frac{1}{2}} \mathbf{h}_m$, $\tilde{\mathbf{H}} =\begin{bmatrix} \tilde{\mathbf{h}}_1&\cdots &\tilde{\mathbf{h}}_M\end{bmatrix}$, $\tilde{\mathbf{n}}_{\rm BS}= \left(\sigma^2\mathbf{I}_N+\alpha P_0\mathbf{C}_{\rm SI}\right)^{-\frac{1}{2}}\left(\mathbf{s}_{\rm SI} +\mathbf{n}_{\rm BS}\right)$, $\mathbf{D} $ is a diagonal  matrix with the elements on its main diagonal as $ {D}_{m,m}=\sqrt{\eta_m}\mathbf{h}_m^T \mathbf{s}_0 $, and $ \mathbf{s}=\begin{bmatrix} s_1&\cdots &s_M \end{bmatrix}^T$. %We note that  the noise covariance matrix is an identity matrix, after applying the noise whitening process, $CN(0,\mathbf{I}_N)$.   Therefore, different SIC strategies can be applied here, where we focus on QR.

By   treating   $\mathbf{D} $ as a power allocation matrix,   the system model in \eqref{overall macx1}  can be viewed as a special case of  a conventional MAC, whose sum capacity can be realized as follows \cite{Goldsmith03}. Without loss of generality, assume that    ${\rm BD}_m$'  signal is   decoded at the $m$-th stage of SIC, where it is straightforward to show that different SIC decoding orders lead to the same sum capacity. Further assume that prior to decoding ${\rm BD}_m$'  signal, the  signals from ${\rm BD}_i$, $1\leq i \leq m-1$, have been decoded correctly.  By applying $\sqrt{\eta_m}\mathbf{h}_m^H \mathbf{s}_0 ^*\tilde{\mathbf{h}}_m^H \left( \mathbf{I}_N+\sum^{M}_{i=m+1} {\eta_i}  |\mathbf{h}_i^T \mathbf{s}_0|^2 \tilde{\mathbf{h}}_i\tilde{\mathbf{h}}_i^H\right)^{-1}$ as the detector, the following data rate is achievable to ${\rm BD}_m$:
 \begin{align}\nonumber
 R_m^{\rm MAC}=&\log\det\left(\mathbf{I}_{N}+\left( \mathbf{I}_N+\sum^{M}_{i=m+1} {\eta_i}  |\mathbf{h}_i^T \mathbf{s}_0|^2 \tilde{\mathbf{h}}_i\tilde{\mathbf{h}}_i^H\right)^{-1} \right.
 \\ 
 &\left.\times {\eta_m}  |\mathbf{h}_m^T \mathbf{s}_0|^2\tilde{\mathbf{h}}_m \tilde{\mathbf{h}}_m ^H\right),
 \end{align}
 for $1\leq m \leq M-1$, and $R_M^{\rm MAC}=\log\det\left(\mathbf{I}_{N}+   {\eta_M}  |\mathbf{h}_M^T \mathbf{s}_0|^2\tilde{\mathbf{h}}_M \tilde{\mathbf{h}}_m ^H\right)$. Therefore,  the sum capacity of the uplink BackCom devices achieved by this type of SIC is given by  
 \begin{align}
 &R_{\rm sum}^{\rm MAC}   = \log\det \left( \mathbf{I}_N +  \sum^{M}_{m=1} {\eta_m}\left|\mathbf{h}_m^T \mathbf{s}_0\right|^2\tilde{\mathbf{h}}_m\tilde{\mathbf{h}}_m^H\right) .
 \end{align}
This letter considers a  problem of  throughput maximization   for the BackCom devices, which can be formulated as follows: 
  \begin{problem}\label{pb:1} 
  \begin{alignat}{2}
\underset{ \eta_m}{\rm{max}} &\quad    
R_{\rm sum}^{\rm MAC}    \label{1obj:1} \\
\rm{s.t.} & \quad  
\sum^{M}_{m=1} {\eta_m} |{g}_{m,k}|^2
  |\mathbf{h}_m^T\mathbf{W}|^2 \leq \tau_k, \quad 1\leq k \leq K\label{1st:1}
\\
& \quad 0\leq \eta_m\leq 1, \quad 1\leq m \leq M \label{1st:2},  \end{alignat}
\end{problem} 
 where  $\tau_k$ indicates the tolerable interference experienced by downlink user ${\rm U}_k$. If ${\rm U}_k$ has a target data rate, denoted by $R_k$, one choice of $\tau_k$ is given by  $\tau_k=\frac{|\mathbf{g}_k^T\mathbf{w}_k|^2 -\epsilon_k\sum^{K}_{i=1,i\neq k}|\mathbf{g}_k^T\mathbf{w}_i|^2 }{\epsilon_k}-\frac{\sigma^2}{P_0},$ and $\epsilon_k=2^{R_k}-1$. 

{\it Remark 1:} Problem \ref{pb:1} is a concave optimization problem since its objective function  is  in a log-det form and its constraints are affine \cite{S0895479896303430}. Various optimization solvers, such as Matlab fmincon, can be straightforwardly applied  to find the optimal solution of  Problem \ref{pb:1}.

{\it Remark 2:} The solution of Problem \ref{pb:1} is based on an instantaneous realization of $x_k$. Therefore,  significant system overhead can be consumed for the base station to inform the BackCom devices about the optimal choices of $\eta_m^*$. A low-complexity alternative is to use   random choices of $\eta_m$ which satisfy \eqref{1st:1} and \eqref{1st:2}. The simulation results provided in Section \ref{section simulation} show that Approach I with random $\eta_m$ can still significantly outperform OMA.

{\it Remark 3:} Another type of complexity introduced  by Approach I is explained in the following. During each SIC step,  $\sqrt{\eta_m}\mathbf{h}_m^H \mathbf{s}_0^* \tilde{\mathbf{h}}_m^H \left( \mathbf{I}_N+\sum^{M}_{i=m+1} {\eta_i}  |\mathbf{h}_i^T \mathbf{s}_0|^2 \tilde{\mathbf{h}}_i\tilde{\mathbf{h}}_i^H\right)^{-1}$ needs to be computed, where the inverse of the matrix requires   a computational complexity of $\mathcal{O}(N^3)$.   In the case of the number of antennas, $N$, is large, the computational complexity for generating  this detector    can be significant, which   motivates the low-complexity  approach  introduced in the next section.  
\vspace{-0.5em}

\subsection{Approach II - A Low-complexity QR Based Deign}\label{subsection 3.2}
Assume that the composite channel matrix $\tilde{\mathbf{H}}$ in \eqref{overall macx1} can be decomposed via QR decomposition as follows: $\tilde{\mathbf{H}}=\mathbf{Q}\mathbf{R}$, where $\mathbf{Q}$ is an $N\times N$ unitary matrix, and $\mathbf{R}$ is an $N\times M$ upper triangular matrix \cite{1285069}. The base station can  use $\mathbf{Q}^H$   as a detection matrix, which simplifies    the system model in \eqref{overall macx1}   as follows:
  \begin{align}\label{system model qr2}
\mathbf{Q}^H\tilde{\mathbf{y}}_{\rm BS}  =  \mathbf{R} \mathbf{D}\mathbf{s} + \mathbf{Q}\tilde{\mathbf{n}}_{\rm BS}. 
 \end{align}
By using the   upper triangular   structure of $\mathbf{R}$, SIC can be implemented in a low-complexity manner \cite{1285069}. In particular, during the $(M-m+1)$-th step, 
the signal from ${\rm BD}_m$ can be decoded with the following data rate:
 \begin{align}
 R_m = \log\left(1+R_{m,m}^2 {\eta_m}|\mathbf{h}_m^T \mathbf{s}_0|^2 \right),
 \end{align}
 where $R_{m,m}$ is defined similar to $D_{m,m}$.  
 To further reduce system overhead, it is ideal to formulate the resource allocation problem based on the following average  sum rate:
  \begin{align}
 \bar{R}_{\rm sum} = \sum^{M}_{m=1}\mathcal{E}_{x_k} \left\{ \log\left(1+R_{m,m}^2 {\eta_m}|\mathbf{h}_m^T \mathbf{s}_0|^2  \right)\right\}.
 \end{align}
 
 Therefore,   the considered   long-term throughput maximization    problem can  be formulated as follows:
 \begin{problem}\label{pb:2} 
  \begin{alignat}{2}
\underset{ \eta_m}{\rm{max}} &\quad    
 \bar{R}_{\rm sum} \label{1obj:2}  \quad
\rm{s.t.} & \quad  \eqref{1st:1}, \eqref{1st:1}.  \end{alignat}
\end{problem} 

We note that there is an analogy between Problem \ref{pb:2} and the one developed for cognitive MAC \cite{5290301}; however,   the explicit expression of $ \bar{R}_{\rm sum}$ can be obtained, as shown in the following.   First define    $\bar{R}_m$ as follows: 
   \begin{align}
 \bar{R}_{m} =  \mathcal{E}_{x_k} \left\{ \log\left(1+R_{m,m}^2 {\eta_m}\mathbf{s}_0^H\mathbf{h}_m^*\mathbf{h}_m^T \mathbf{s}_0  \right)\right\}.
 \end{align}
 Recall that $ \mathbf{s}_0 = \sqrt{P_0}\sum^{K}_{k=1}\mathbf{w}_kx_k$, which means that  $ \bar{R}_{m}$ can be rewritten as follows:
   \begin{align}\label{xxxx3}
 \bar{R}_{m} =  \mathcal{E}_{x_k} \left\{ \log\left(1+ P_0R_{m,m}^2 {\eta_m} \mathbf{x} ^H\mathbf{W}^H\mathbf{h}_m^*\mathbf{h}_m^T  \mathbf{W}\mathbf{x}  \right)\right\}.
 \end{align} 
A closed-form expression of $ \bar{R}_{m} $ can be found by using   the method developed in \cite{6517941}: 
    \begin{align}
 \bar{R}_{m} =&  \log(e)\mathcal{E}_{x_k} \left\{\int^{\infty}_{0} \left(\frac{e^{-t}}{t} \right.\right.
 \\\nonumber 
 &\left. \left.-\frac{1}{t} e^{ -t\left(1+P_0R_{m,m}^2 {\eta_m} \mathbf{x} ^H\mathbf{W}^H\mathbf{h}_m^*\mathbf{h}_m^T  \mathbf{W}\mathbf{x}  \right)}\right)dt\right\},
 \end{align} 
which can be simplified as follows:
      \begin{align}
 \bar{R}_{m} =&\log(e)\int^{\infty}_{0} \left(\frac{e^{-t}}{t} \right. 
 \\\nonumber 
 &\left.  -\frac{e^{-t}}{t} \mathcal{E}_{x_k} \left\{e^{ -tP_0 R_{m,m}^2 {\eta_m} \mathbf{x} ^H\mathbf{W}^H\mathbf{h}_m^*\mathbf{h}_m^T  \mathbf{W}\mathbf{x}   }\right\}\right)dt.
 \end{align} 
 Assuming that  $x_k$'s are independent and identically distributed (i.i.d.) complex Gaussian variables with mean zero and unit variance,  $ \bar{R}_{m}$ can be evaluated as follows: 
       \begin{align}\nonumber
 \bar{R}_{m} =& \log(e)\int^{\infty}_{0} \left(\frac{e^{-t}}{t}    -\frac{e^{-t}}{t \left(1+t P_0R_{m,m}^2 {\eta_m} |\mathbf{h}_m^T  \mathbf{W}|^2  \right)} \right)dt\\
\label{x3}
 =&
 \log(e) f\left( P_0R_{m,m}^2 {\eta_m} |\mathbf{h}_m^T  \mathbf{W}|^2\right), 
 \end{align}
 where $f(x)\triangleq  -e^{\frac{1}{x}}E_i\left(- \frac{1}{x}\right)$ and the last step follows from \cite[3.352.6]{GRADSHTEYN} and $E_i(\cdot)$ denotes the exponential integral function. We note that an alternative way to obtain \eqref{x3} is to treat $\mathbf{h}_m^T  \mathbf{W}\mathbf{x} $ in \eqref{xxxx3} as a complex Gaussian variable with mean zero and variance $|\mathbf{h}_m^T  \mathbf{W}|^2$. 
 
 So Problem \ref{pb:2} can be recast in the following equivalent form:  
  \begin{problem}\label{pb:3} 
  \begin{alignat}{2}
\underset{\eta_m}{\rm{max}} &\quad    
\sum^{M}_{m=1} -e^{\frac{1}{ P_0R_{m,m}^2 {\eta_m} |\mathbf{h}_m^T  \mathbf{W}|^2}} E_i\left(-\frac{1}{  P_0R_{m,m}^2 {\eta_m} |\mathbf{h}_m^T  \mathbf{W}|^2}\right) \label{1obj:3} \\
\rm{s.t.} & \quad  \eqref{1st:1}, \eqref{1st:1}. 
  \end{alignat}
\end{problem} 
Although Problem \ref{pb:3} contains the exponential integral function, it is still  concave   as shown in  the following lemma.
\begin{lemma}\label{lemma1}
  Problem \ref{pb:3} is a concave optimization problem. 
\end{lemma}
\begin{proof}
See Appendix \ref{proof1}.
\end{proof}

  {\it Remark 4:} Because Problem  \ref{pb:3}  is concave, it can be straightforwardly solved by using various optimization solvers. We note that the computation of the   function $f(x)$ in \eqref{x3} can be difficult even for a moderately small $x$. For example, for $x=0.0013$, $\frac{1}{x}=750$, and Matlab returns $f(750)={\rm Inf}$. To overcome this computational issue, the following approximation of $f(x)$ is used for small $x$. 
 Recall that $-E_i\left(- \frac{1}{x}\right)=\Gamma\left(0,\frac{1}{x}\right)$, where $\Gamma(\beta, y)$ denotes the incomplete gamma function and can be approximated as follows \cite{GRADSHTEYN}: 
 \begin{align}
 \Gamma(\beta, y) \approx y^{\beta-1}e^{-y} \sum^{L-1}_{m=0}\frac{(-1)^m\Gamma(1-\beta+m)}{y^m\Gamma(1-\beta)},
 \end{align}
 for $y\rightarrow \infty$. 
 By letting $L=1$ and $
 \beta=0$, $ f(x) $ can be approximated as follows:
  \begin{align}
f(x)=&  -e^{\frac{1}{x}} \Gamma\left(0,\frac{1}{x}\right)  \approx    x  , \quad x\rightarrow 0,
 \end{align}
which can be used to approximate   \eqref{1obj:3}.
 
 {\it Remark 5:} Compared to Approach I, the QR based design can be implemented with  low computational complexity, as explained in the following. First, there is no need to calculate the inverse of a matrix with size of $N$ at each SIC step. Second, the resource allocation solution is not based on the instantaneous realizations of $x_k$, which reduces the system complexity.  However,  it is worth to point out that, unlike Approach I,    Approach II cannot achieve the sum capacity of MAC, or support   the overloading case, i.e., $M>N$.

 \section{ Extension to OFDMA-Based Legacy Systems } \label{section 4}
 The aforementioned BAC-NOMA scheme can also be applied to the case,  where OFDMA is used in the legacy  system. In particular, in the considered legacy system, a base station serves  $K$ downlink users, each denoted by ${\rm U}_k$, $1\leq k \leq K$, via $K$ orthogonal OFDMA subcarriers. If OMA is used, all the $K$ subcarriers will be occupied by a single BackCom device, since   a signal  reflected by one BackCom device can block  all subcarriers.    The application of BAC-NOMA can ensure that    $M$ uplink BackCom devices  are simultaneously admitted to share the subcarriers. For   the purpose of illustration,  it is assumed that each node is  equipped with a single antenna.  
 
Without loss of generality, assume that    downlink user ${\rm U}_k$ is served at  the $k$-th subcarrier.  
Following the  ambient BackCom model in \cite{8476159,8446004}, at subcarrier $k$,  the frequency-domain baseband signal received by  downlink user ${\rm U}_k$ is given by
  \begin{align}\nonumber
y_k^{\rm D} =  \sqrt{P_0}{G}_{k} x_k+\sqrt{P_0}\sum^{M}_{m=1}\sqrt{\eta_m} {G}_{m,k}H_{m,k}x_k  s_m  + {n}^{\rm D}_{k},
 \end{align}
 where $ {G}_k$ denotes the channel gain between the base station and   ${\rm U}_k$ at subcarrier $k$, $ {H}_{m,k}$ denotes the forward channel gain from the base station to uplink device ${\rm BD}_m$ at subcarrier $k$, $G_{m,k}$ denotes the channel gain between ${\rm U}_k$ and  ${\rm BD}_m$ at subcarrier $k$, and   $n_k^D$ denotes the receiver noise.

Because $s_m$ is unknown to the downlink users, the term, $I_o\triangleq \sqrt{P_0}\sum^{M}_{m=1}\sqrt{\eta_m} {G}_{m,k}^kH_{m,k}x_k  s_m $, is again treated  as interference. The   power of this interference term  is given by $
\mathcal{E}_{x_k,s_m} \left\{I_oI_o^*\right\}
= P_0  \sum^{M}_{m=1} {\eta_m} |{G}_{m,k}|^2|H_{m,k}|^2 $.
 Therefore, the data rate of downlink user ${\rm U}_k$ is given by
 \begin{align}
 &R_{k}^{\rm D} = \log\left( 1+\frac{P_0| {G}_k|^2  }{P_0\sum^{M}_{m=1} {\eta_m} |{G}_{m,k}|^2|H_{m,k}|^2   +\sigma^2}
 \right).
 \end{align}
 
At the base station, the frequency-domain baseband signal at the $k$-th subcarrier is given by 
 \begin{align}\label{overall macz}
 {y}^{\rm BS}_k = \sqrt{P_0}\sum^{M}_{m=1}\sqrt{\eta_m}{F}_{m,k} {H}_{m,k}  x_k s_m +  {s}^{\rm SI}_k + {n}^{\rm BS}_k,
 \end{align}
where  $F_{m,k}$ denotes ${\rm BD}_m$'s backward channel gain  at subcarrier $k$,   ${s}^{\rm SI}_k $ denotes the self-interference and $ {n}^{\rm BS}_k$ denote the noise. As in  the previous section, it is assumed that   ${s}^{\rm SI}_k  \sim CN(0, \alpha P_0 |h_{\rm SI}^k|^2)$, where $h_{\rm SI}^k$ denotes  the self-interference channel. Furthermore, it is assumed that self-interference at different subcarriers is  independent.

By   applying the pre-whitening process, the system model at the base station can be expressed  as follows:
 \begin{align}\label{overall macz2}
  \breve{y}^{\rm BS}_k =\sqrt{P_0} x_k  \breve{\mathbf{h}}_k^H  \boldsymbol \eta^{\frac{1}{2}} \mathbf{s}  +   \breve{n}^{\rm BS}_k,
 \end{align}
 where $ \breve{y}^{\rm BS}_k  =  \left(\alpha P_0 |h_{\rm SI}^k|^2+\sigma^2\right)^{-\frac{1}{2}} {y}^{\rm BS}_k $, $\breve{\mathbf{h}}_k=\left(\alpha P_0 |h_{\rm SI}^k|^2+\sigma^2\right)^{-\frac{1}{2}}\begin{bmatrix}{F}_{1,k} {H}_{1,k}&\cdots &{F}_{M,k} {H}_{M,k} \end{bmatrix}^H$, $\mathbf{s}=\begin{bmatrix}{s}_{1}  &\cdots &{s}_{M }   \end{bmatrix}^T$, ${\boldsymbol \eta}$ is an $M\times M$ diagonal matrix, i.e., ${\boldsymbol \eta}=\text{diag}\{\eta_1, \cdots, \eta_M\} $, and $\breve{n}^{\rm BS}_k$ is a complex Gaussian white noise with mean zero and unit variance. Stacking the $K$ observations in one vector, the system model at the base station can be rewritten as follows:
  \begin{align}\label{overall macz3}
 \breve{ \mathbf{y}}^{\rm BS} = \sqrt{P_0} \mathbf{D}_x \breve{\mathbf{H}}^H  \boldsymbol \eta^{\frac{1}{2}} \mathbf{s}   +\breve{\mathbf{n}}^{\rm BS},
 \end{align}
 where $\mathbf{D}_x $ is an $K\times K$ diagonal matrix, i.e., $\mathbf{D}_x=\text{diag}\{x_1, \cdots, x_K\} $, $\breve{\mathbf{H}} = \begin{bmatrix}\breve{\mathbf{h}}_1 &\cdots &\breve{\mathbf{h}}_K \end{bmatrix}$, $  \breve{\mathbf{y}}^{\rm BS}= \begin{bmatrix}\breve{y}^{\rm BS}_1 &\cdots &\breve{y}^{\rm BS}_K \end{bmatrix}^T$,  and $\breve{\mathbf{n}}^{\rm BS}$ is constructed similarly to $ \breve{\mathbf{y}}^{\rm BS}$. 
 
By defining $\bar{\mathbf{H}}= \sqrt{P_0}\mathbf{D}_x \breve{\mathbf{H}}^H $ and treating $\boldsymbol \eta^{\frac{1}{2}} $ as a power allocation matrix,  one can view   the system model in \eqref{overall macz3} as a special case of   conventional MAC, whose sum capacity can be realized as follows.  Without loss of generality, assume that    ${\rm BD}_m$'  signal is   decoded at the $m$-th stage of SIC, where it is  noted that  that different SIC decoding orders lead to the same sum capacity. Further assume that prior to decoding ${\rm BD}_m$'  signal, the  signals from ${\rm BD}_i$, $1\leq i \leq m-1$, have been decoded correctly.  By applying $\bar{\mathbf{h}}_m^H\left( \mathbf{I}_K+\sum^{M}_{i=m+1} {\eta_i}    \bar{\mathbf{h}}_i\bar{\mathbf{h}}_i^H\right)^{-1} $ as the detector, the following data rate is achievable to ${\rm BD}_m$:
 \begin{align}\nonumber
 &R_m^{\rm MAC}=\\\nonumber &\log\det\left(\mathbf{I}_{K}+\left( \mathbf{I}_K+\sum^{M}_{i=m+1} {\eta_i}    \bar{\mathbf{h}}_i\bar{\mathbf{h}}_i^H\right)^{-1}   {\eta_m}  \bar{\mathbf{h}}_m \bar{\mathbf{h}}_m ^H\right),
 \end{align}
 for $1\leq m \leq M-1$, and $R_M^{\rm MAC}=\log\det\left(\mathbf{I}_{K}+   {\eta_M}  \bar{\mathbf{h}}_M \bar{\mathbf{h}}_m ^H\right)$, where $\bar{\mathbf{h}}_i$ denotes the $i$-th column of $ \bar{\mathbf{H}}$. Therefore,  the sum capacity  of the uplink BackCom devices  is given by 
 \begin{align}
 &R_{\rm sum}^{\rm MAC}  =\log\det \left( \mathbf{I}_K + \sum^{M}_{m=1}   {\eta_m}  \bar{\mathbf{h}}_m \bar{\mathbf{h}}_m ^H\right) ,
 \end{align}
 which can be used to formulate the following throughput maximization problem: 
  \begin{problem}\label{pb:5} 
  \begin{alignat}{2}
\underset{ \eta_m}{\rm{max}} &\quad    
R_{\rm sum}^{\rm MAC}    \label{1obj:5} \\
\rm{s.t.} & \quad  
   \sum^{M}_{m=1} {\eta_m} |{G}_{m,k}|^2|H_{m,k}|^2\leq \tau_k, \quad 1\leq k \leq K \\\nonumber &\quad \eqref{1st:2}.  \end{alignat}
\end{problem}    
Similar to   \eqref{1obj:1},   \eqref{1obj:5} is also in the concave log-det form.   Therefore, Problem \ref{pb:5} is also concave and hence can be straightforwardly solved.

 \begin{figure}[t]\centering \vspace{-0.5em}
    \epsfig{file=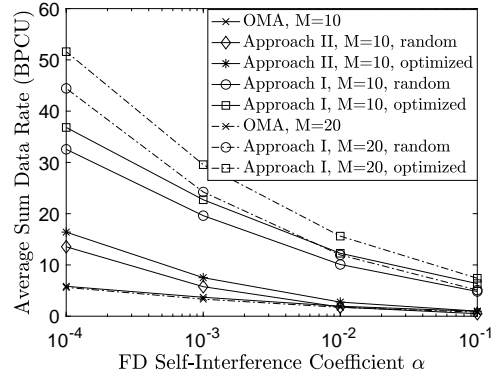, width=0.35\textwidth, clip=}\vspace{-0.5em}
\caption{ The impact of FD self-interference on the performance of the considered transmission schemes. BPCU denotes bits per channel use.  Case I is used for the locations of the downlink users  and $K=N=10$.      }\label{fig1}\vspace{-1.5em}
\end{figure}
\vspace{-1em}
\section{Simulation Results}\label{section simulation}
In this section, simulation results are presented to evaluate the performance of the proposed BAC-NOMA schemes.  For   the presented results,    the path loss exponent is set as $3$,   $P_0=20$ dBm,  $\tau_k=0.01$, $\sigma_n^2=-94$ dBm, and the base station is located at $(0,0)$ m.    The uplink BackCom devices are randomly located inside a square of side $6$ m whose center is located at $(0, 0)$ m. Two cases, termed Cases I and II, are considered for the locations of the downlink users. For Case I, the downlink users are   randomly located inside a square of side $6$ m with its center located at   $(0, 0)$ m, whereas for case II, the users are also   in the same size square with its center at  $(3, 0)$ m. 

In Fig. \ref{fig1}, the impact of FD self-interference on the performance of the considered transmission schemes is studied. As can be seen from the figure, the NOMA schemes can outperform the OMA scheme, since multiple uplink BackCom devices can be simultaneously supported by the NOMA schemes. Furthermore, the figure shows that Approach I outperforms Approach II and also is applicable to the overloading cases, i.e., $M>N$. However,   it is worth to point out that the performance gain of Approach I is obtained at a price of more system complexity, as discussed in Remarks 2 and 3.   In addition,  Fig. \ref{fig1} shows  that increasing $\alpha$ decreases the performance of all schemes, which is due to the fact that a larger $\alpha$ means more residual FD self-interference and hence leads to more performance degradation.

In Fig. \ref{fig2}, the impact of the number of BackCom devices on the performance of the considered schemes is studied. With more BackCom devices participating in NOMA transmission, the performance of the BAC-NOMA schemes is improved, whereas the impact of   $M$ on the performance of OMA is insignificant. Another important observation from Fig. \ref{fig2} is that the performance of the schemes for Case II is better than that of Case I, which is due to the fact that the interference between the uplink and downlink users is more severe    in Case I.   As discussed in Section \ref{section 4}, the concept of BAC-NOMA can be straightforwardly extended to the case with OFDMA based legacy systems, which is demonstrated in Fig. \ref{fig3}. Because $K$ subcarriers are used in the considered legacy system, the normalized sum-rate, i.e., $\frac{R_{\rm sum}^{\rm MAC}}{K}$, is used as the metric for the performance evaluation. As can be observed from Fig. \ref{fig3}, the BAC-NOMA schemes can outperform OMA, and inviting more BackCom devices to participate  in NOMA transmission improves the throughput, which are consistent to the observations made in Figs. \ref{fig1} and \ref{fig2}.

 \begin{figure}[t]\centering \vspace{-2.5em}
    \epsfig{file=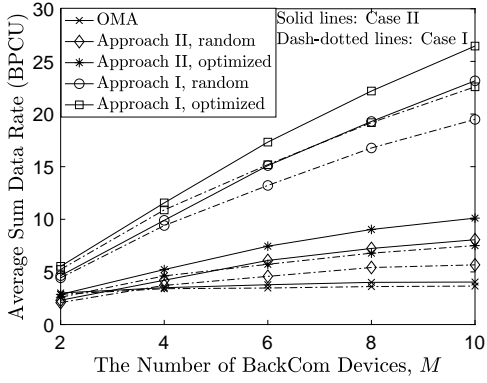, width=0.35\textwidth, clip=}\vspace{-0.5em}
\caption{ The impact of the number of BackCom devices, $M$, on  the performance of the considered transmission schemes.   $\alpha=0.001$.       }\label{fig2}\vspace{-1.5em}
\end{figure}

\vspace{-1em}
\section{Conclusions}
This letter has demonstrated the advantages   of  NOMA for ambient BackCom networks with  OFDMA and SDMA used  in legacy systems. By using BAC-NOMA, multiple BackCom devices can be served simultaneously, instead of being served in different time slots as in OMA.  Two   resource allocation approaches have been proposed in order to realize   different tradeoffs between system performance and complexity.   
\vspace{-1em}

\appendices 
  \section{Proof for Lemma \ref{lemma1}}\label{proof1}
 Recall that all the constraints of Problem \ref{pb:3}  are affine, and therefore the lemma can be proved by showing that the objective function of Problem \ref{pb:3}  is concave.   
Recall that $f(x)=  -e^{\frac{1}{x}}E_i\left(- \frac{1}{x}\right)$, $x\geq0$. The concavity of the objective function in \eqref{1obj:3} can be proved by showing that $f(x)$ is a concave function. Recall that the first order derivative of  $f(x)$ has been obtained in \cite{bacnomamtc} as follows:
\begin{align}  
f'(x)  \label{first orderdd}
=&  e^{x^{-1}} x^{-2}E_i\left(- x^{-1}\right) +x^{-1}.
\end{align}
By using \eqref{first orderdd} and with  some algebraic manipulations, the second order derivative $f''(x)$ can be expressed as follows:
\begin{align}   \label{first orderdd3}
&f''(x)=  e^{x^{-1}} x^{-4}\\\nonumber
&\times \left( -  E_i\left(- x^{-1}\right) -xe^{-x^{-1}}   -2x E_i\left(- x^{-1}\right)    -x^2 e^{-x^{-1}}\right). 
\end{align}
Define $u(x) =  -  E_i\left(- x^{-1}\right) -xe^{-x^{-1}}   -2x E_i\left(- x^{-1}\right)    -x^2 e^{-x^{-1}}$. In order to show $f''(x)\leq 0$, it is sufficient to show that $u(x)\leq 0$, since $x\geq 0$.
The first order derivative of $u(x)$ is given by
\begin{align}  
u'(x) \nonumber
=&   e^{-x^{-1}}x^{-1} -e^{-x^{-1}} - e^{-x^{-1}}x^{-1}   -2 E_i\left(- x^{-1}\right)  \\\nonumber&   +2x  e^{-x^{-1}}x^{-1} -2x e^{-x^{-1}}   - e^{-x^{-1}} 
\\\label{first orderdd32}
=&        -2 E_i\left(- x^{-1}\right)      -2x e^{-x^{-1}}     ,
\end{align}
where the first step follows from the fact that $\frac{dE_i\left(-x^{-1}\right)}{dx} = -\frac{e^{-x^{-1}}}{x}$. 

Define $g(x) \triangleq  E_i\left(- x^{-1}\right) +xe^{-x^{-1}}$, and  $u'(x)$ can be shown as a function of $g(x)$ as follows: 
\begin{align}  
u'(x)  \label{first orderdd33}
=&        -2 \left( E_i\left(- x^{-1}\right)      +x e^{-x^{-1}}  \right)   =-2g(x).
\end{align}
In \cite{bacnomamtc}, it is shown that $g(x)\geq 0$ for $x\geq 0$. Therefore, $u'(x) \leq 0$ for $x\geq 0$, i.e., $u(x)$ is a monotonically decreasing function of $x$, which leads to the following inequality:
\begin{align}\nonumber
u(x)&= -  E_i\left(- x^{-1}\right) -xe^{-x^{-1}}   -2x E_i\left(- x^{-1}\right)    -x^2 e^{-x^{-1}}\\   &\leq u(0)=0.
\end{align}
Because $u(x)$ is non-positive, the second order derivative of $f(x)$ is also non-positive.
Therefore, $f(x)$ is concave and hence the objective function of Problem \ref{pb:3} is also concave since a non-negative weighted sum of concave functions is still concave.  The proof for the lemma is complete. 

  \begin{figure}[t]\centering \vspace{-2.5em}
    \epsfig{file=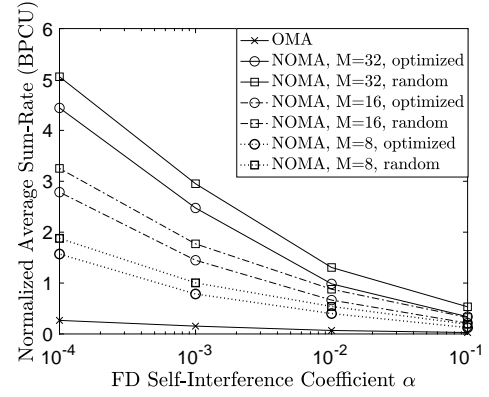, width=0.34\textwidth, clip=}\vspace{-0.5em}
\caption{ The performance of BAC-NOMA in OFDMA based legacy systems.   Case I is used for the locations of the downlink users.  $N=K=16$.      }\label{fig3}\vspace{-1.5em}
\end{figure}
\vspace{-1em}
\bibliographystyle{IEEEtran}
\bibliography{IEEEfull,trasfer}

  \end{document}